\documentclass{llncs}

\usepackage{multicol}

\usepackage{amsmath}
\usepackage{amsfonts}
\usepackage{amssymb}

\usepackage{graphicx}
\usepackage{epic}
\usepackage{eepic}
\usepackage{epsfig,float}
\usepackage{verbatim}
\usepackage{pdfsync}

\newcommand{\ol}{\overline}
\newcommand{\eps}{\varepsilon}
\newcommand{\emp}{\emptyset}

\newcommand{\Sig}{\Sigma}
\newcommand{\sig}{\sigma}
\newcommand{\noin}{\noindent}

\newcommand{\bi}{\begin{itemize}}
\newcommand{\ei}{\end{itemize}}
\newcommand{\be}{\begin{enumerate}}
\newcommand{\ee}{\end{enumerate}}
\newcommand{\bd}{\begin{description}}
\newcommand{\ed}{\end{description}}
\newcommand{\bq}{\begin{quote}}
\newcommand{\eq}{\end{quote}}
\newcommand{\txt}[1]{\mbox{ #1 }}

\newcommand{\floor}[1]{\lfloor #1 \rfloor}

\newcommand{\cD}{{\mathcal D}}

\newtheorem{cor}[theorem]{Corollary}

\newtheorem{conj}[theorem]{Conjecture}

\title{Syntactic Complexity of Finite/Cofinite, Definite, and Reverse Definite Languages
\thanks{This work was supported by the Natural Sciences and Engineering Research Council of Canada under grant No.~OGP0000871
}
}

\author{Janusz~Brzozowski and David Liu
 }

\authorrunning{Brzozowski, Liu}   

\institute{David R. Cheriton School of Computer Science, University of Waterloo, \\
Waterloo, ON, Canada N2L 3G1\\
{\tt \{brzozo, dyliu\}@uwaterloo.ca}
}

\begin{document}

\maketitle
%\today
\begin{abstract}
We study the syntactic complexity of finite/cofinite, definite and reverse definite languages.
The syntactic complexity of a class of languages is defined as the maximal size of  syntactic semigroups of languages from the class, taken as a function of the state complexity $n$ of the languages. 
We prove that $(n-1)!$ is a tight upper bound for finite/cofinite languages and that it can be reached only if the alphabet size is greater than or equal to 
$(n-1)!-(n-2)!$. We prove that the bound is also $(n-1)!$ for reverse definite languages, but the alphabet size is $(n-1)!-2(n-2)!$.
We show that $\lfloor e\cdot (n-1)!\rfloor$ is a lower bound on the syntactic complexity of definite languages, and  conjecture that this is also an upper bound, and  that the alphabet size required to meet this bound is $\floor{e \cdot (n-1)!} - \floor{e \cdot (n-2)!}$.
We prove the conjecture for $n\le 4$.
\medskip

\noin
{\bf Keywords:}
definite, finite automaton, finite/cofinite, regular language, reverse definite, syntactic complexity, syntactic semigroup
\end{abstract}

\section{Introduction}

A language is \emph{definite} if it can be decided whether a word $w$ belongs to it simply by examining 
the suffix of $w$ of  some fixed length.
The class of definite languages was the very first subclass of regular languages to be considered: it was introduced in 1954 in the classic paper by Kleene~\cite{Kle54}. 
It was then studied in 1963 by Perles, Rabin, and Shamir~\cite{PRS63}, and Brzozowski~\cite{Brz63}, in 1966 by Ginzburg~\cite{Gin66}, 
and later by several others.
Definite languages were  revisited  in 2009 by Bordihn, Holzer and Kutrib~\cite{BHK09} in connection with state complexity. 
\emph{Reverse definite} languages were first studied by Brzozowski~\cite{Brz63}. 
Here membership of $w$ can be determined by  its prefix of some fixed length. 
The class of \emph{finite and cofinite} languages is the intersection of the definite and reverse definite classes.
Here testing for membership can be done by checking all words shorter than some fixed length.
These three classes appear at the bottom of the dot-depth hierarchy~\cite{BrSi73} of star-free languages, below generalized definite languages and locally testable languages. 
All three classes are boolean algebras. 
The semigroup $S$ of a finite/cofinite language is \emph{nilpotent}: It has a single idempotent $e$ which is a zero, and is characterized  by the equations $eS=Se=e$.
For definite (reverse definite) languages every idempotent $e$ is a right zero, that is, $Se=e$ (respectively, a left zero, that is, $eS=e$).

We study the sizes of syntactic semigroups of finite/cofinite, definite, and reverse definite  languages.
If $L\subseteq\Sig^*$ is a regular language over alphabet $\Sig$, its syntactic semigroup is defined by the Myhill congruence~\cite{Myh57} $\approx_L $: For $x,y\in\Sig^*$,
\begin{equation*}
x \approx_L y \mbox{ if and only if } uxv\in L  \Leftrightarrow uyv\in L\mbox { for all } u,v\in\Sig^*.
\end{equation*}
The set $\Sig^+/ \approx_L$ of equivalence classes of the relation $\approx_L$ is  the \emph{syntactic semigroup} of $L$. 
It is well-known that this semigroup is isomorphic to the semigroup $T_L$ of transformations performed by non-empty words in  the minimal deterministic finite automaton (DFA) recognizing $L$~\cite{McPa71}, and it is usually convenient to deal with the latter semigroup.
It is obvious that the transformation semigroup of the minimal DFA of $L$ is identical to that of the minimal DFA of $\ol{L}$, the complement of $L$.

The \emph{syntactic complexity} $\sig(L)$ of a language $L$ is the size of its syntactic semigroup, and $\sig(L)=|T_L|$, where $|S|$ denotes the cardinality of a set $S$.
Syntactic complexity can vary significantly among languages with the same state complexity~\cite{BrYe11}, where the \emph{state complexity} of a language is the number of states in its minimal DFA.

The observation that $n^n$ is a tight upper bound on the size of the transformation semigroup of a DFA with $n$ states was first made by Maslov~\cite{Mas70} in 1970, although this follows immediately from a 1935 result of Piccard~\cite{Pic35}, who showed that three generators suffice to produce all transformations of a set of $n$ elements. 
The interest in syntactic complexity of \emph{subclasses} of regular languages is new. 
In 2003--2004 Holzer and K\"onig~\cite{HoKo04}, and  Krawetz, Lawrence and Shallit~\cite{KLS03} studied \emph{unary} and \emph{binary} languages. 
In 2011 Brzozowski and Ye~\cite{BrYe11} showed the following bounds:
\emph{right ideals}---tight upper bound $n^{n-1}$; \emph{left ideals}---lower bound $n^{n-1}+n-1$; 
\emph{two-sided ideals}---lower bound $n^{n-2}+(n-2)2^{n-2}+1$.
In 2012 Brzozowski, Li and Ye~\cite{BLY11} found the following bounds:
\emph{prefix-free} languages---tight upper bound $n^{n-2}$; \emph{suffix-free} languages---lower bound 
$(n-1)^{n-2}+n-2$;
\emph{bifix-free} languages---lower bound 
$(n-1)^{n-3}+(n-2)^{n-3}+(n-3)2^{n-3}$; \emph{factor-free} languages---lower bound 
$(n-1)^{n-3}+(n-3)2^{n-3}+1$.
Also in 2012 tight upper bounds were found for three subclasses of star-free languages by Brzozowski and Li~\cite{BrLi11}:
\emph{monotonic} languages---$C_n^{2n-1}$; \emph{partially monotonic} 
languages---$f(n)=\sum_{k=0}^{n-1} C^{n-1}_k C^{n+k-2}_k$;
\emph{nearly monotonic} languages---$f(n)+n-1$, where $C^i_j$ is the binomial coefficient
\emph{$i$ choose $j$}.
It was conjectured in~\cite{BrLi11} that the bound for nearly monotonic languages is also a tight upper bound for star-free languages.
That bound is asymptotically $2^{-3/4}{(\sqrt{2}+1)^{2n-1}}/{\sqrt{\pi (n-1)}}$.

We prove that $(n-1)!$ is a tight upper bound for finite/cofinite languages, and that a growing alphabet of size at least
$(n-1)!-(n-2)!$ is required to reach the bound.
For reverse definite languages the bound is also $(n-1)!$, but the alphabet size is now 
$(n-1)!-2(n-2)!$.
We  show that $\lfloor e\cdot (n-1)!\rfloor$ is a lower bound for definite languages, and that it can be reached with an alphabet of size $\lfloor e\cdot (n-1)!\rfloor - \lfloor e\cdot (n-2)!\rfloor$.
We conjecture that this is also an upper bound, and 
prove the conjecture  for $n\le 4$.

There is a lack of left-right symmetry in several results for syntactic complexity in spite of the fact that the syntactic congruence is symmetric.
Thus, in the case of ideals~\cite{BrYe11}, it was easy to find a tight upper bound for right ideals, but no tight upper bound is known for left ideals. It was easy to find a tight upper bound for prefix-free languages, but no tight upper bound is known for suffix-free languages~\cite{BLY11}.
This happens again here. We have a tight upper bound for reverse definite languages, but  no tight upper bound for definite languages. 

Section~\ref{sec:prelims}
contains some preliminary material. Sections~\ref{sec:FC}--\ref{sec:DEF}
discuss the syntactic complexity of finite/cofinite, reverse definite, and definite languages, respectively, and Section~\ref{sec:conc} concludes the paper.

\section{Preliminaries}\label{sec:prelims}

A {\em transformation} of a set $Q$ is a mapping of $Q$ into itself. We consider only transformations of finite sets, and assume without loss of generality  that $Q = \{1,2,\ldots, n\}$. If $t$ is a transformation of $Q$, 
and  $i \in Q$, then $it$ is the image of $i$ under $t$.  
An arbitrary transformation can be written in the form
\begin{equation*}\label{eq:transmatrix}
t=\left( \begin{array}{ccccc}
1 & 2 &   \cdots &  n-1 & n \\
i_1 & i_2 &   \cdots &  i_{n-1} & i_n
\end{array} \right),
\end{equation*}
where $i_k = kt$, $1\le k\le n$, and $i_k\in Q$. We also use the notation $t = [i_1,i_2,\ldots,i_n]$ for the transformation $t$ above. 

If $X$ is a subset of $Q$, then $Xt = \{it \mid i \in X\}$, and the {\em restriction} of $t$ to $X$, denoted by $t|_X$, is a mapping from $X$ to $Xt$ such that $it|_X = it$ for all $i \in X$. 

A \emph{permutation} of $Q$ is a mapping of $Q$ \emph{onto} itself. 
A transformation $t$ is \emph{permutational} if there exists some $X \subseteq Q$ with $|X| \ge 2$ such that $t|_X$ is a permutation of $X$. Otherwise, $t$ is \emph{non-permutational}.

A \emph{constant} transformation, denoted by $Q \choose j$, has $it=j$ for all $i$.

The {\em composition} of two transformations $t_1$ and $t_2$ of $Q$ is a transformation $t_1 \circ t_2$ such that $i (t_1 \circ t_2) = (i t_1) t_2$ for all $i \in Q$. We usually omit the composition operator. 

A~\emph{deterministic finite automaton} (DFA) is a quintuple $\cD=(Q, \Sig, \delta, q_1,F)$, where 
$Q$ is a finite, non-empty set of \emph{states}, $\Sig$ is a finite non-empty \emph{alphabet}, $\delta:Q\times \Sig\to Q$ is the \emph{transition function}, $q_1\in Q$ is the \emph{initial state}, and $F\subseteq Q$ is the set of \emph{final states}. We extend $\delta$ to $Q \times \Sig^*$ in the usual way.
The DFA $\cD$ \emph{accepts} a word $w \in \Sigma^*$ if ${\delta}(q_1,w)\in F$. 
The set of all words accepted by $\cD$ is the language $L(\cD)$ of $\cD$. 
Two states of a DFA are \emph{distinguishable} if there exists a word $w$
which is accepted from one of the states and rejected from the other. Otherwise,
the two states are \emph{equivalent}.
A DFA is \emph{minimal} if all of its states are reachable from the initial state and no two states are equivalent. 
All the minimal DFA's of a given language $L$ are isomorphic. 

The notion of a DFA $\cD$ connects transformations to regular languages. Given a regular language $L$, its minimal DFA $\cD=(Q,\Sig, \delta,q_1, F)$, and a word $w \in \Sigma^+$, the transition function $\delta(\cdot, w)$  is a \emph{transformation} of $Q$, the {transformation caused by $w$}. When convenient, we identify a word with its corresponding transformation.

The (\emph{left}) \emph{quotient} of a language $L\subseteq \Sig^*$ by a word $w\in\Sig^*$ is the language $L_w=\{x\mid wx\in L\}$. 
Note that $L_\eps=L$, where $\eps$ is the empty word.
The \emph{quotient DFA} of a regular language $L$ is 
$\cD=(Q, \Sig, \delta, q_1,F)$, where $Q=\{L_w\mid w\in\Sig^*\}$, $\delta(L_w,a)=L_{wa}$, 
$q_1=L_\eps=L$,  and $F=\{L_w\mid \eps \in L_w \}$.
The quotient DFA is isomorphic to the minimal DFA accepting $L$.

\section{Finite/Cofinite Languages} \label{sec:FC}

One of the simplest classes of regular languages is the class of finite and cofinite languages, where a language is \emph{cofinite} if its complement is finite. 
Since the syntactic complexity bounds for finite and cofinite languages are identical,  we restrict our analysis here to finite languages.

Let $L$ be a regular language and $\cD = (Q, \Sigma, \delta, q_1, F)$ be its minimal DFA. 
It is well-known that $L$ is finite/cofinite if and only if there exists a numbering $1,\dots, n$ on $Q$ so that for all $w \in \Sigma^*$, $\delta(i,w) = j$ implies that $i < j$ or $i = j = n$.
We define the set $A_n$ of transformations on $\{1,2, \dots, n\}$ with these properties: 
$$A_n = \{t \mid it > i \; \forall \; i = 1,\dots,n-1, \txt{and} nt = n\}.$$
It is clear that $A_n$ is a semigroup under composition of size $(n-1)!$.

\begin{theorem} 
\label{thm:FC}
Let $L$ be a finite or cofinite language with state complexity $n$. Then the syntactic complexity of $L$ satisfies $\sigma(L) \le (n-1)!$ and this bound is tight.
\end{theorem}
\begin{proof}
Let $\cD = (Q, \Sigma, \delta, q_1, F)$ be the minimal DFA of $L$. 
The above discussion implies that we may label the states $Q$ so that $T_L$ is a subsemigroup of $A_n$.
Therefore the bound holds.

Let $n \ge 1$ and $|\Sigma| = (n-1)!$. Let $\cD$ be a $DFA$ with states numbered $\{1,2,\dots,n\}$, initial state $1$, sink state $n$, and a final state $n-1$. For each transformation $t \in A_n$, assign a letter in $\Sigma$ whose input transformation on $\cD$ is exactly~$t$.
To show that $\cD$ is minimal, note that state $i>1$ is reached from the initial state by the transformation $[i, n,n,\ldots,n]$. Also, if $i$ and $j$  are two states and $i < j \le n$, then the transformation $t \in A_n$ that has $it = n-1$, and $kt = n$ for all other $k \neq i$, distinguishes the two states. Hence $\cD$ is minimal and accepts a finite language. 
Therefore the bound is tight. 
\qed
\end{proof}

A natural question is the minimal size of the alphabet required to achieve the upper bound. 
Let $\cD$ be the minimal DFA of a finite language $L$ with $T_L = A_n$.
For any state $i \in Q$ and $a \in \Sigma$, it is clear that $\delta(i,a) \geq i+1$ or $i = n$.
It follows that if an input transformation $t \in A_n$ satisfies $it = i+1$ for \emph{some} $i \in \{1,2,\dots,n-2\}$, then any word $w$ corresponding to $t$ must have length 1, that is, $w$ must be in $\Sigma$.

\begin{theorem} 
\label{thm:finAlpha}
Let $L \subseteq \Sigma^*$ be a finite or cofinite language with state complexity~$n\ge 3$, and suppose that $\sigma(L) = (n-1)!$. Then $$|\Sigma| \ge (n-1)! - (n-2)!$$ and this bound is tight.
\end{theorem}
\begin{proof} 
By Theorem \ref{thm:FC}, we may assume that $T_L = A_n.$ The preceding discussion implies that $|\Sigma|$ is at least the number of transformations which satisfy $it = i + 1$ for some $i = 1,\dots, n-2$. Let $G_n \subset A_n$ be the set of these transformations. If we place the restriction $it \neq i + 1$ for all $i \in \{1,2,\dots,n-2\}$ then there are $n -i - 1$ choices for these $it$, and hence a total of $(n-2)!$ such transformations. Therefore $|G_n| = |A_n| - (n-2)! = (n-1)! - (n-2)!.$
Now let $t = [j_1, \dots, j_{n-2}, n, n] \in A_n$ be arbitrary. Let 
$$k = \min_{1 \le i \le n-2} \{j_i - i\} - 1,$$ and $t'=[j_1-k,\ldots,j_{n-2}-k,n,n].$
Then  $t'\in G_n$ and $t = t'[2,3, \dots, n-1, n, n]^k $. Thus $G_n$ generates $A_n$, and the bound is tight. \qed
\end{proof}

\begin{example}
\label{ex:FC}
For $n=4$, the largest semigroup is 
$$A_4=\{{\bf [2,3,4,4]},{\bf [2,4,4,4]},{\bf [3,3,4,4]},[3,4,4,4],{\bf [4,3,4,4]},[4,4,4,4]\},$$
and its minimal generating set is shown in boldface.
\end{example}

\section{Reverse Definite Languages}
\label{sec:RD}

A \emph{reverse definite language} is a language $L \subseteq \Sigma^*$ of the form $L = E \cup F\Sigma^*$, where $E$ and $F$ are finite languages.
Because reverse definite languages are characterized by prefixes of a fixed length, their minimal DFAs (and hence syntactic complexity bounds) are very similar to those of finite/cofinite languages. 
If $L$ has state complexity 1, then either $L=\emp$ or $L=\Sig^*$. Since both these languages are in the finite/cofinite class, the bound $(n-1)!$ of Theorem~\ref{thm:FC} applies. 
For state complexities $n > 1$, we note first that if $\emp$ is not a quotient of $L$, then $L$ is cofinite.
Otherwise, $\emp$ and $\Sigma^*$ are both quotients of $L$. 
Let $\cD = (Q, \Sigma, \delta, q_1, F)$ be the minimal DFA of $L$, and label the states corresponding to $\emp$ and $\Sigma^*$ with $n-1$ and $n$, respectively. 
One can number the other states in $Q$ so that for all words $w \in \Sigma^*$, if $\delta(i, w) = j$ then $i \leq j$ with equality if and only if $i \in \{n-1,n\}$.

The syntactic complexity results for reverse definite languages now follow directly from the finite/cofinite results.

\begin{theorem} Let $L = E \cup F \Sigma^*$ be a reverse definite language with state complexity $n \ge 1$. Then $\sigma(L) \le (n-1)!,$ and this bound is tight. Moreover, if this 
language $L$ achieves this upper bound and $n \ge 4$, then $|\Sigma| \ge (n-1)! - 2(n-2)!$, and this bound is tight.
\end{theorem}
\begin{proof}
First, if $\emp$ is not a quotient of $L$, then $L$ is cofinite and hence has the same bounds as in the previous section. 
To find a cofinite witness $L$ meeting the bound $(n-1)!$, first find a finite witness $\ol{L}$ as in the proofs of Theorems \ref{thm:FC}
and~\ref{thm:finAlpha}, and then interchange its final and non-final states.

Otherwise, let $\cD$ be the minimal  DFA recognizing $L$, and let the states be totally ordered as in the preceding discussion. Define the set of transformations analogous to the finite case: 
$$A_n ' = \{t \mid it > i \; \forall\, i = 1, \dots, n-2, \text{ $(n-1)t = n-1$, and $nt = n$}\}.$$
Then $T_L \subseteq A_n '$, which a straightfoward calculation shows to be a semigroup.  
Clearly, $|A_n'| = (n-1)!$, thus proving the bound.

To find a witness for this case, start with the finite witness as in the proofs of Theorems \ref{thm:FC}
and~\ref{thm:finAlpha}, make all transitions from state $n-1$ to go to itself, and make state
$n$ the only final state.

For the minimal size of the alphabet, we define $G'_n \subset A'_n$ to be the set of transformations $t$ in $A'_n$ satisfying $it = i + 1$ for some $i = 1,\dots, n-3$. 
As in Section \ref{sec:FC}, these transformations must correspond to individual letters in $\Sigma$, hence proving the bound.
The same indirect counting argument shows that for $n \geq 4$, $|G'_n| = (n-1)! - 2 \cdot (n-2)!.$
A similar argument also shows that $G'_n$ generates $A'_n$ (using the transformations $[2,3,\dots, n-1, n-1, n]$ and $[2,3,\dots, n, n-1, n]$ in place of $[2,3,\dots, n-1, n,n]$).
Therefore the alphabet size bound is tight.
\qed
\end{proof}

\begin{example}
\label{ex:RD}
For $n=4$, the finite witness meeting the bound $(n-1)!$ has the transformation set
given in Example~\ref{ex:FC}.
We modify this set by making $n-1$ the sink state, thus obtaining
$$A'_4=\{{\bf [2,3,3,4]},{\bf [2,4,3,4]},[3,3,3,4],[3,4,3,4],[4,3,3,4],[4,4,3,4]\},$$ 
where the generators are in boldface, and state 4 is final.
\end{example}

\section{Definite Languages} \label{sec:DEF}
A \emph{definite language} is a language $L \subseteq \Sigma^*$ of the form $L = E \cup \Sig^*F$, where $E$ and $F$ are finite languages.
Like finite/cofinite and reverse definite languages, definite languages are characterized by their transformation semigroups.
In this case, every transformation of the minimal DFA of a regular language must be non-permutational. 
Conversely, if the transformation semigroup of a minimal DFA contains only non-permutational transformations, then it accepts a definite language.

Our goal for this section is to find the maximal size of a non-permutational transformation semigroup, that is, one which contains only non-permutational transformations.
There is a straightforward bijection between such transformations on $\{1,\dots, n\}$ and simple labeled forests on $n-1$ nodes. 
This can be seen by constructing the graph on $n$ nodes with edges $ij$ representing $it = j$, and then removing the unique node for which $it = i$.
Then Cayley's Theorem~\cite{Cay89,Sho95} shows that there are $n^{n-1}$ non-permutational transformations of $\{1,\dots, n\}$.

Identifying non-permutational transformations is not sufficient to find a syntactic complexity bound, as the set of such transformations does not form a semigroup for $n \ge 3$. 
For example, the composition of $s=[2,3,3]$ and $t=[1,1,2]$ is $st=[1,2,2]$, which is permutational.
Two transformations \emph{conflict} if there exists a permutational transformation in the semigroup that they generate.

We exhibit the following sets of non-permutational transformations which do not conflict; they are similar to the semigroup $A_n$ from Section \ref{sec:FC}.

\begin{theorem}\label{thm:def} Let $n > 1$, and define the following sets of transformations: 
$$ B_{n,k} = \{t \mid it > i \; \forall \; 1 \le i < k, \text{ and } \; it = k \; \forall \;  i \ge k\}, \quad k = 1, 2, 3, \dots, n. $$
Then the set of transformations $\displaystyle{B_n = \bigcup_{k=1}^n B_{n,k}}$ is a maximal non-permutational semigroup of size $\floor{e \cdot (n-1)!}$.
\end{theorem}
\begin{proof}
One can check that each $B_{n,k}$ is a semigroup. Let $t_i \in B_{n,i}$ and $t_j \in B_{n,j}$, with $i < j$. A direct computation shows that $t_it_j \in B_{n,it_j}$, and $t_jt_i \in B_{n,i}$; hence $B_n$ is a semigroup. Moreover, for all $t \in B_{n,k}$, $t^{k-1} = \binom{Q}{k}$, and so all of the transformations are non-permutational.

A simple counting argument shows that 
$$|B_{n,k}| = (n-1)(n-2)\cdots (n - k + 1) = \frac{(n-1)!}{(n-k)!}.$$ 
Since the $B_{n,k}$ are disjoint,
$$|B_n| = \sum_{k=1}^n \frac{(n-1)!}{(n-k)!} \nonumber \\
= \sum_{l=0}^{n-1} \frac{(n-1)!}{(n-1-l)!} \nonumber \\
= \floor{e \cdot (n-1)!}. $$

For the maximality of $B_n$, we show that adding any other non-permutational transformation creates a conflict. Let $t \notin B_n$ be non-permutational, with $it = i$. 

First suppose that there exists a $j < i$ with $jt = k \le j$. Since $t$ is non-permutational, we may assume $k < j$. Then there exists a $t' \in B_{n,i}$ with $kt' = j$; then $itt' = i$ and $jtt' = j$, and  so $t$ and $t'$ conflict.

If no such $j$ exists, then  there must exist a $j > i$ with $jt \neq i$. Consider the sequence defined by $j_0 = j$, $j_l = j_{l-1} t$. If there exists an $l$ such that $j_lt = j_{l+1} < i$, let $l$ be the minimal one. Let $t' \in B_{n,j_l}$ with $j_{l+1}t' = i$ and $it' = j_l$. Then $itt' = j_l$,  $j_ltt' = i$, and so $tt'$ is permutational. 
Now suppose all $j_l \geq i$. Since $t$ is non-permutational, $i$ must appear in the sequence; moreover, since $j_1 = jt \neq i$, we can pick $l \geq 0$ so that $i = j_{l+2}$. Since $j_{l+1} > i$, we may find a transformation $t' \in B_{j,j_l}$ with $it' = j_{l+1}$ and $j_{l+1} t' = j_l$.
Then $it't = i$,  $kt't = k$, and $t't$ is permutational. 
\qed
\end{proof}

To compute the generators of $B_n$, we require the following definition.
Let $C_n$ be the set of all transformations $t = [i_1,\dots, i_n] \in B_n$ with all $i_j < n$.
Define the function $\alpha:C_n\to B_n$  by $\alpha(t) = [i_1 + 1,\dots, i_n + 1]$, and also 
$$\alpha(C_n) = \{t \in B_n \mid \alpha(t_0) = t \text{ for some $t_0 \in C_n$}\}.$$
Clearly, $\alpha$ is a bijection between $C_n$ and $\alpha(C_n)$.

\begin{theorem}
Let $H_n = B_n \backslash \alpha(C_n)$. Then 
\begin{enumerate}
\item[(1)] $H_n$ is the minimum set of generators for $B_n$. 
\item[(2)] $|H_n| = \floor{e \cdot (n-1)!} - \floor{e \cdot (n-2)!}.$
\end{enumerate}
\end{theorem}
\begin{proof}
For (1), note that $[2,3,\dots, n,n] \in H_n$. For any $t \in B_n$, we can write $t = t_0 [2,3,\dots, n,n]^k$ with $k \geq 0$ and $t_0 \in H_n$, as in the proof of Theorem \ref{thm:finAlpha}. Therefore $H_n$ generates $B_n$.

Now let $t_i \in B_{n,i}$ and $t_j \in B_{n,j}$, with $i \ge j$. We consider $mt_it_j$, and use the fact that each transformation $t \in B_{n,k}$ satisfies $mt \geq \min \{k, m + 1\}$. There are two cases:
\begin{enumerate}
\item[(a)] If $m \ge j - 1$, then $mt_i \ge \min \{i, m + 1\} \ge j$, hence $mt_i t_j = j$.
\item[(b)] If $m \le j - 2 < i$, then $mt_i \ge m + 1$, hence $mt_i t_j \ge \min \{j, mt_i + 1\} \geq m + 2$.
\end{enumerate}
It follows that $\alpha^{-1}(t_it_j) \in B_{n, j-1}$; a similar argument shows that $\alpha^{-1}(t_jt_i) \in B_{n, jt_i - 1}$. 
Consequently, no transformation in $H_n$ is a composition of two others in $B_n$, and so $H_n$ is the minimum generating set of $B_n$.

For (2), we calculate $|\alpha(C_n)|$, or equivalently $|C_n|$ because $\alpha$ is a bijection. A counting argument shows that $|B_{n,k} \cap C_n| = \frac{(n-2)!}{(n - 2 - (k-1))!}$. Therefore
$$|H_n| = |B_n| - |\alpha(C_n)| = |B_n| - \sum_{k=1}^{n-1} \frac{(n-2)!}{(n-2 - (k-1))!}
= \floor {e \cdot (n-1)!} - \floor {e \cdot (n-2)!}. $$
\qed
\end{proof}

The following corollary establishes a direct connection with definite languages.

\begin{cor} For all $n > 1$, there exists a definite language $L$ with state complexity $n$, syntactic complexity $\sigma(L) = \floor{e \cdot (n-1)!},$ and alphabet size $\floor{e \cdot (n-1)!} - \floor{e \cdot (n-2)!}.$
\end{cor}
\begin{proof}
Let $\cD = (Q, \Sig, \delta, q_1, F)$ be a DFA with $Q = \{1,2,\dots,n\}$, $q_1 = 1$, $F = \{n\}$, and $|\Sigma| = \floor {e \cdot (n-1)!} - \floor {e \cdot (n-2)!}$ with each letter representing a different transformation in $H_n$, so that the transformation semigroup of $\cD$ is $B_n$. We claim that this is a minimal DFA of a definite language. First, all the states are reachable by the constant transformations $\binom{Q}{i} \in B_n$. Also, any two states $i,j$ with $i < j < n$ are distinguishable by the transformation $t \in B_n$ which acts as $kt = k + 1$ for $1 \le k \le i$, and $kt = n$ for $k > i$. State $n$ is distinguishable from every other state because it is the only final state. Hence $\cD$ is minimal. Then by Theorem \ref{thm:defEquiv}, $\cD$ accepts a definite language.
\qed
\end{proof}

\begin{conj} Let $L$ be a definite language with state complexity $n > 1$. Then $\sigma(L) \le \floor{e \cdot (n-1)!},$ and if equality holds then $|\Sigma| \ge \floor{e \cdot (n-1)!} - \floor{e \cdot (n-2)!}.$
\end{conj}

\begin{example}
\label{ex:D}
For $n=4$ we have the following transformations in $B_n$:
\begin{align}
B_{4,1} &= \{{\bf [1,1,1,1]}\}, \nonumber\\
B_{4,2} &= \{ [2,2,2,2], \bf{[3,2,2,2]}, {\bf [4,2,2,2]} \},  \nonumber\\
B_{4,3} &= \{{\bf[2,3,3,3]},{\bf [2,4,3,3]},[3,3,3,3],{\bf [3,4,3,3]}, [4,3,3,3],{\bf [4,4,3,3]}\}, \nonumber\\
B_{4,4} & =\{{\bf [2,3,4,4]},{\bf [2,4,4,4]}, {\bf [3,3,4,4]}, [3,4,4,4],{\bf [4,3,4,4]},[4,4,4,4]\}.
 \nonumber
\end{align}

The generators are shown in boldface.

\end{example}

\section{Conclusions and Future Work} \label{sec:conc}
Though we have found tight upper bounds on the syntactic complexity of finite/cofinite and reverse definite languages, 
we have only conjectured  the bounds on the syntactic complexity and the corresponding alphabet size  for definite languages. The conjecture has been verified through computational enumeration for $n \le 4$, but remains unproven for $n > 4$. Also, syntactic complexity bounds have yet to be found for the related higher classes in the dot-depth hierarchy of star-free languages, namely the generalized definite and locally testable languages. It is possible that the technique used in this paper---characterize allowable transformations in the syntactic semigroup and apply combinatorial arguments to count them---can be used to find bounds for these languages as well.

\providecommand{\noopsort}[1]{}

\end{document}